\documentclass[sigconf]{acmart}
\AtBeginDocument{%
  }

\setcopyright{acmlicensed}
\setcopyright{acmlicensed}
\copyrightyear{2026}
\acmYear{2026}
\setcopyright{cc}
\setcctype{by}
\acmConference[RecSys '26]{20th ACM Conference on Recommender Systems}{September 27-October 02, 2026}{Minneapolis, MN, USA}
\acmBooktitle{20th ACM Conference on Recommender Systems (RecSys '26), September 27-October 02, 2026, Minneapolis, MN, USA}
\acmDOI{10.1145/3773078.3831736}
\acmISBN{979-8-4007-2284-4/2026/09}

\usepackage{booktabs}
\usepackage{multirow}
\usepackage{subcaption}
\usepackage{amsthm}
\theoremstyle{remark}
\newtheorem{remark}{Remark}
\begin{document}

\title{Addressing Cross-Stage Decoupling of Semantic and Collaborative Signals in Generative Recommendation}

\author{Jiayi Dan}
\authornote{Corresponding author.}
\email{danjiayi05@kuaishou.com}
\affiliation{%
  \institution{Kuaishou Technology}
  \city{Beijing}
  \country{China}
}




\renewcommand{\shortauthors}{Jiayi Dan}

\begin{abstract}
Generative recommendation reformulates sequential recommendation as autoregressive generation by encoding items into semantic tokens, enabling improved scaling capability and cross-domain generalization. However, existing generative recommender systems typically follow a two-stage pipeline, where item tokenization is largely dominated by textual semantics with limited incorporation of collaborative signals and interaction similarity, leading to code assignments that are misaligned with downstream generation. Conversely, the generation stage tends to overlook the original semantic information, as the code sequences are re-embedded based on interaction data. This cross-stage information decoupling limits semantic coherence and recommendation accuracy.

To address this issue, we propose SCRec, a general framework that enhances cross-stage coherence through bidirectional information supplementation. Specifically, we introduce (i) collaborative-enhanced tokenization to explicitly inject textualized collaborative signals into semantic tokenization, without introducing additional alignment task, (ii) semantic-guided generation to dynamically recalibrate semantic priors with learnable code embeddings in generation stage, and (iii) manifold alignment to reconcile the geometric mismatch between the embedding space of discrete codebook indices and the dense continuous semantic space. These interrelated components form a general framework that aligns semantic and collaborative signals and enhances cross-stage information coherence, with minimal additional training and inference costs. Extensive experiments \footnote{Code is available at \url{https://github.com/DanJiayi/SCRec}.} demonstrate the effectiveness, robustness, and generalizability of our proposed framework.
\end{abstract}




\maketitle

\section{Introduction}
Generative recommendation has emerged as a promising paradigm that reformulates sequential recommendation as autoregressive generation. Prior work shows that representing items as semantic tokens before generation improves recommendation performance and cross-domain generalization through richer semantic abstraction. In the tokenization stage, item textual information is encoded by a pretrained language model (PLM) and quantized into discrete code sequences to capture hierarchical semantics while leveraging the rich knowledge of PLMs. During generative training, these code sequences are re-embedded based on interaction data to capture behavioral and collaborative patterns flexibly.

However, as illustrated in Figure ~\ref{fig:intro}, this two-stage paradigm introduces an information decoupling issue. The tokenization stage is typically dominated by semantic information, making it difficult to coherently inject collaborative signals and interaction similarity, resulting in code assignments that are not well suited for downstream generation. In contrast, the generation stage tends to overlook the original semantic information, whereas the next item is expected to be decoded as a semantic code sequence, leading to semantic distortion and reduced generation accuracy.

To address this semantic-collaborative decoupling issue, we propose a general and lightweight framework that improves cross-stage coherence via bidirectional information supplementation. On one hand, we introduce collaborative-enhanced tokenization to explicitly inject textualized collaborative signals into code sequences without introducing additional alignment task. On the other hand, we propose semantic-guided generation to dynamically fuse textual priors with codebook embeddings during generation. However, code embeddings and semantic representations differ not only in modality but also in their geometry: the former originates from discrete codebook indices, while the latter are dense continuous hidden states of language models. To address this underlying discrepancy, we introduce hyperbolic manifold alignment to enforce geometric consistency. These interconnected components ensure consistent use of semantic and collaborative signals across stages.

We also implement the proposed framework as a reusable module that can be integrated into a wide range of generative recommenders with minimal training and inference overhead.

Our contributions are summarized as follows:

\begin{itemize}
\item We identify the semantic–collaborative decoupling issue in existing two-stage generative recommenders and integrate collaborative-aware tokenization with semantic-guided generation to enhance cross-stage information coherence.
\item We introduce a manifold alignment mechanism to geometrically bridge token embedding space and dense semantic representations in a shared hyperbolic latent space.
\item Based on these designs, we develop a reusable module that can be incorporated into existing generative recommendation models with minimal additional overhead.
\item We conduct extensive experiments to validate the effectiveness and generality of the proposed framework.
\end{itemize}

\begin{figure}[htbp]
    \centering
    \includegraphics[width=\linewidth]{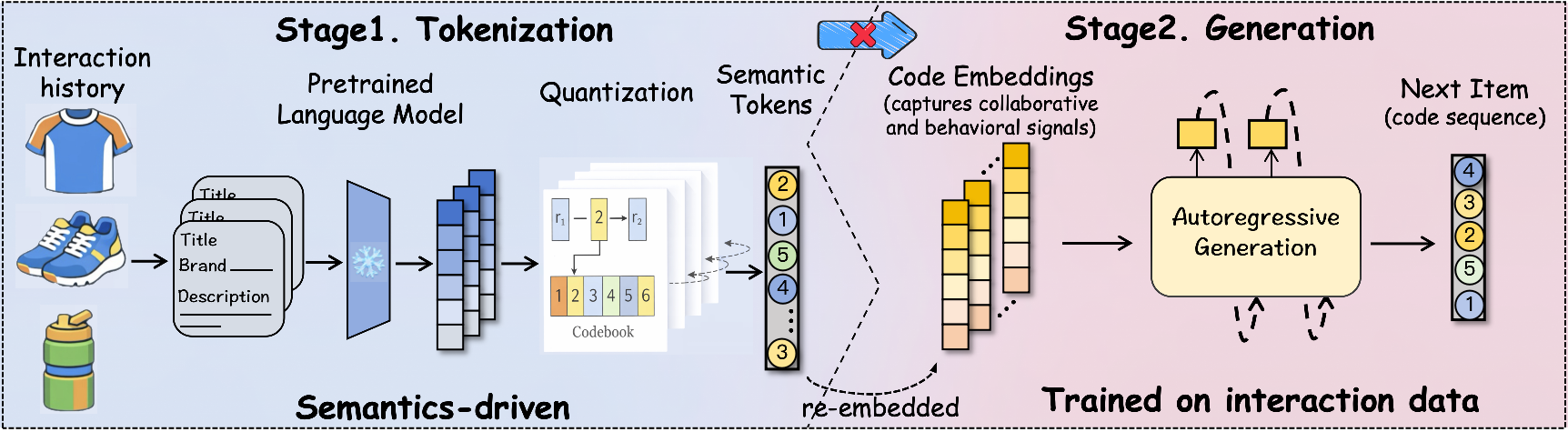}
    \caption{Illustration of Semantic ID-Based Two-Stage Generative Recommendation.}
    \label{fig:intro}
\end{figure}


\section{Related Work}
Sequential recommendation is traditionally formulated as a discriminative prediction task with large and sparse item embedding tables \cite{kang2018sasrec,sun2019bert4rec}. Recent work reframes it as sequence generation, where models autoregressively decode item identifiers or natural language descriptions \cite{rajput2023tiger, geng2022p5, bao2023tallrec, li2023gpt4rec}. Language models are widely used for contextual modeling and generation, including encoder-decoder architectures \cite{geng2022p5,chu2023leveraging,cui2022m6} and decoder-only architectures \cite{bao2023tallrec,ji2024genrec,zhou2025generative}.

Beyond the classic autoregressive paradigm, alternative generation paradigms have been explored to improve efficiency and robustness. RPG \cite{hou2025generating} and SETRec \cite{lin2025order} reformulate item identifier prediction as parallel generation, reducing decoding latency and error propagation. CAR \cite{wang2025act} introduces chunk-level modeling to capture users’ decision-making processes. Diffusion-based approaches \cite{shi2025llada,li2024recdiff,zhao2024denoising,song2025diffcl,li2025dimerec,liu2025diffgrm} enable parallel refinement of item identifiers while modeling global dependencies.

Apart from backbone choice and generation paradigms, a key challenge is how items are represented for generation. Sparse ID-based identifiers assign each item a unique numeric token \cite{geng2022p5,zhai2024actions}, which is simple but semantically opaque and memory-intensive; textual identifiers improve alignment with language models \cite{cui2022m6,lee2025gram}, but are less flexible and may suffer from hallucination \cite{zhang2026unleashing}. A dominant paradigm derives semantic IDs by discretizing continuous item representations into sequences of discrete codes \cite{rajput2023tiger, zheng2024adapting, wang2024letter,lin2025unified, zhong2025pctx}. Typically, RQ-VAE constructs coarse-to-fine representations \cite{rajput2023tiger,wang2024letter}; UTGRec \cite{zheng2025universal} models shared and incremental semantics via tree-structured codebooks; RPG \cite{hou2025generating} leverages Product Quantization to build ultra-long semantic IDs. To mitigate token collisions, \cite{lin2025unified} introduces a low-dimensional code representation, while \cite{fang2025hid} separates colliding items via a uniqueness loss. GenCDR \cite{hu2026ids} further extends semantic ID tokenization to cross-domain scenarios by injecting domain-specific signals via LoRA-based adapters into a universal RQ-VAE encoder.

However, reconstruction-centric objectives of semantic-ID tokenization may be misaligned with recommendation tasks \cite{liang2026rethinking}. Recent studies address this mismatch through joint optimization of tokenization and generation, including self-improving identifiers \cite{chen2024enhancing}, end-to-end frameworks with recommendation-oriented signals \cite{liu2025generative}, co-evolution learning \cite{wang2026pit}, and differentiable SID that update codes through gradients of the recommendation loss \cite{fu2026differentiable}.

Another challenge in tokenization for generative recommendation is integrating multi-source data, including multi-modal and collaborative signals. To incorporate multi-modal information, recent work aligns heterogeneous modalities into a unified representation space \cite{zhai2025multimodal,zhu2025beyond,zheng2025universal,wang2025generative,zhai2025simple,zhang2026multi}. Typically, MACRec \cite{zhang2026multi} promotes consistency between text and image modality tokens through implicit contrastive losses and explicit cross-modal generation tasks. Another line of work focuses on incorporating collaborative signals during tokenization. MMGRec \cite{liu2024mmgrec} extracts collaborative filtering (CF) embeddings using a pretrained model and concatenates them with semantic embeddings before quantization. However, such direct concatenation may disrupt the spatial structure of representation vectors. PRORec \cite{xiao2025progressive} aligns the two embeddings via contrastive learning and introduces an extra next-item prediction task during the tokenization stage. However, the alignment quality is not guaranteed, and it incurs noticeable training overhead since it introduces an additional training stage.
Another approach treats CF embeddings as auxiliary supervision rather than direct inputs during quantization. For example, LETTER \cite{wang2024letter} applies contrastive learning between CF embeddings and reconstructed semantic embeddings during RQ-VAE training, while COSETTE \cite{lepage2025closing} encourages the quantized space to approximate the item similarity structure derived from the co-occurrence matrix. However, these implicit approaches take only semantic embeddings as inputs, making it difficult to ensure that collaborative signals are effectively injected. Moreover, these methods fail to recover the rich semantic information encoded in code sequences in the generation stage.

\textbf{Research Gap}: Overall, existing generative recommendation frameworks typically adopt a two-stage pipeline, where the tokenization stage lacks explicit and coherent incorporation of collaborative signals, while the generation stage tends to overlook the original semantic information. Integrating and aligning the two signals to enhance cross-stage information coherence remains a crucial yet underexplored challenge.

\section{Method}
In this section, we present SCRec, a general and lightweight framework to address the cross-stage decoupling between semantic and collaborative signals in generative recommendation. 

Given a user interaction sequence $s = \{i_1, i_2, \dots, i_{t-1}\}$, the task is to predict the next item $i_t$ by generating its corresponding semantic code sequence. Our proposed method explicitly incorporate textualized collaborative signals into the semantic tokenization process, restore the original semantic priors in the generation stage, and further reconcile the intrinsic geometries of the two representations in a shared hyperbolic space. The overall framework is illustrated in Figure~\ref{fig:model}.

\begin{figure*}[t]
    \centering
    \includegraphics[width=0.78\linewidth]{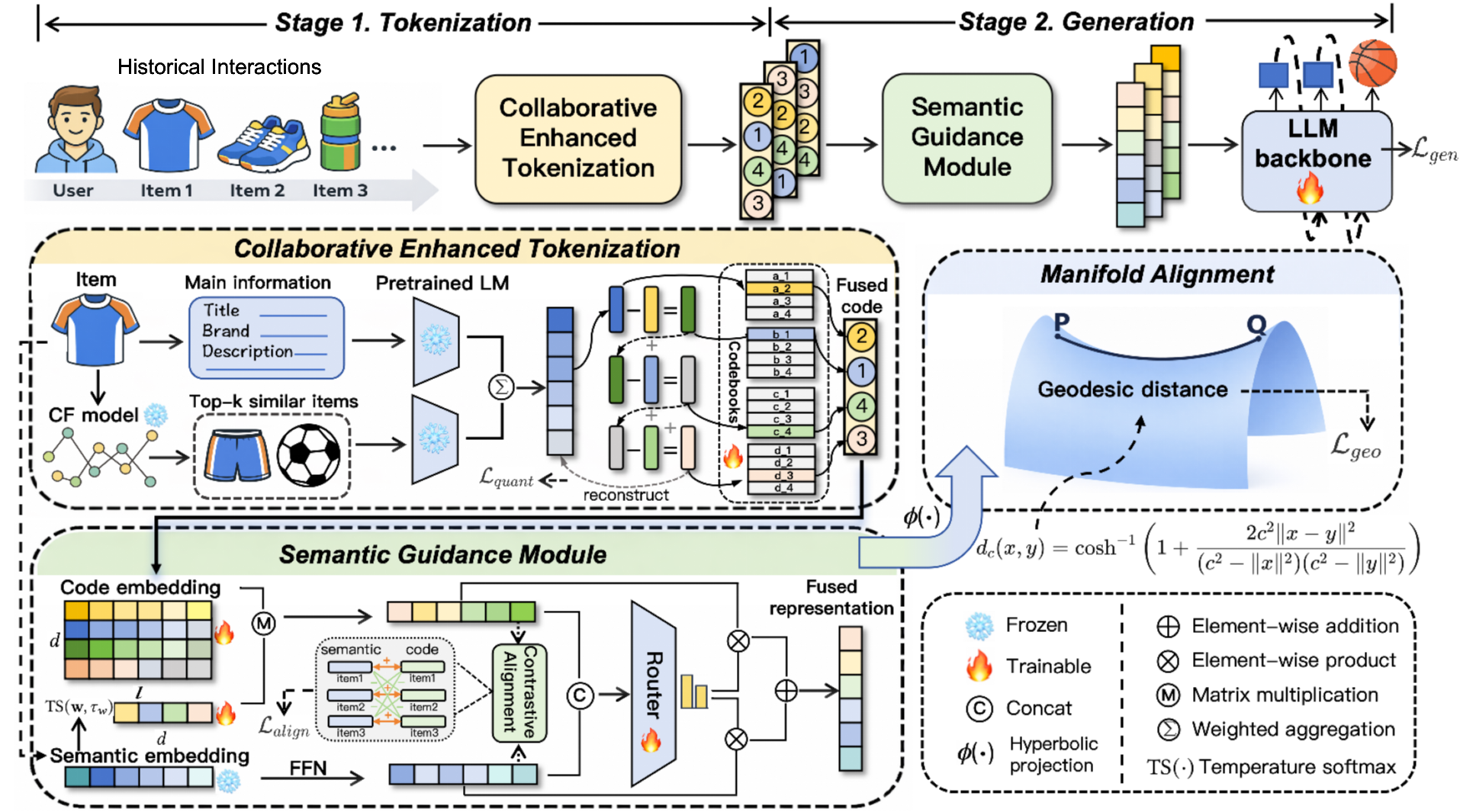}
    \caption{Overview of the proposed framework}
    \label{fig:model}
\end{figure*}

\subsection{Collaborative-Enhanced Tokenization}

The primary objective of tokenization is to represent items as discrete code sequences for downstream generation while preserving semantic structure and leveraging the rich knowledge of pretrained language models. Prior work has shown the effectiveness of representing items as hierarchical semantic code sequences with quantization techniques, typically RQ-VAE. However, existing frameworks often fail to effectively incorporate collaborative signals and interaction similarity during tokenization, leading to a mismatch between code assignment and downstream generation task.

Consider the example of running shoes and high-heeled shoes: since they share similar textual semantics, a purely semantic-driven tokenizer will assign them closely related code sequences.
However, generative models may struggle to learn optimal code embeddings from interactions, as their collaborative patterns differ significantly while the codes enforce similarity. Conversely, items with strong collaborative similarity, such as running shoes and soccer balls, may receive drastically different code sequences due to textual differences, even though they should be close in the generation embedding space, further complicating the downstream task.

To mitigate this issue, recent studies have attempted to incorporate collaborative signals into the tokenization stage. A typical approach concatenates collaborative embeddings obtained from pretrained models with semantic embeddings before quantization \cite{liu2024mmgrec}, but such direct fusion may distort the representation space due to modality mismatch. Another line of work  uses collaborative signals as additional supervision during quantization \cite{wang2024letter,lepage2025closing}. However, this implicit strategy cannot guarantee effective learning of auxiliary objectives and sufficient capture of collaborative signals, as the input consists solely of semantic embeddings.

To inject collaborative signals explicitly and coherently, we transform collaborative signals into textual form, based on the principle that items with similar collaborative and interaction patterns should have similar code assignments.
Specifically, we first employ a collaborative filtering (CF) model to identify the top-$k$ most similar items for each item $i$.
This CF model can be arbitrary, such as a simple Swing algorithm.
Alternatively, to further improve alignment with downstream generation, item similarities can be computed directly from the embeddings of a simple sequential recommendation model pretrained on past interaction data:  
\[
\{ j_1,j_2 \ldots, j_k \}
=
\operatorname*{Top\text{-}k}_{j \in \mathcal{I}} \, \mathrm{sim}(\mathbf{e}_i, \mathbf{e}_j),
\]
where $\mathcal{I}$ denotes the item set,  $\mathbf{e}_i$ denotes the item embedding  obtained from the pretrained model, $\text{sim}(\cdot)$ denotes cosine similarity.

Next, we extract semantic information from both item $i$ and its top-
$k$ most similar items. For item $i$ itself, we preserve fine-grained semantic attributes and organize them as follows:
\[
\text{Prompt}^{\text{main}}_i =
\text{``}\{\text{title}_i\}; \{\text{brand}_i\}; \{\text{category}_i\}; \{\text{description}_i\};\dots\text{''}.
\]
For similar items, we use their titles as the primary semantic cues, concatenate them into a single sequence, and truncate each title to the first 20 words to avoid being dominated by overly long texts:
\[
\text{Prompt}^{\text{CF}}_i =
\text{``}\{\text{title}_{j_1}\}; \{\text{title}_{j_2}\}; \dots; \{\text{title}_{j_k}\}\text{''}.
\]

We then encode the two semantic inputs with a pretrained language model (e.g., T5), obtain their hidden representations, and perform a weighted 
aggregation to derive a unified embedding:

\[
\tilde{\mathbf{e}}_i = a\, \text{Encode}(\text{Prompt}^{\text{main}}_i) + (1-a)\text{Encode}(\text{Prompt}^{\text{CF}}_i) ,
\]
where $a$ is a tunable hyperparameter. 

Note that both semantic and collaborative representations are extracted from textual descriptions using the same pretrained language model, so they are naturally aligned in the same space in a simple yet effective manner, while fully leveraging the knowledge of the pretrained language model. Items with similar interaction patterns and shared neighbors are assigned more similar code sequences, thereby better aligning with the downstream generation task.
The procedure can be performed offline (offline extraction of item embeddings using a pretrained recommender model is a common step in prior work \cite{liu2024mmgrec,wang2024letter,xiao2025progressive,lin2025order}), incurring no additional online training 
or inference cost. 

Finally, we employ RQ-VAE to quantize $\tilde{\mathbf{e}}_i$ into a unified code sequence that jointly captures semantic and collaborative information, and compute the quantization loss:

\begin{equation*}
\mathcal L_{\text{quant}}
=
\sum_{i \in \mathcal{I}}
\Big(
\left\| \tilde{\mathbf e}_i - \hat{\mathbf e}_i \right\|_2^2
+
\sum_{m=0}^{l-1}
\big(
\left\| \operatorname{sg}(\mathbf r_{i,m}) - \mathbf q_{i,m} \right\|_2^2 
+
\gamma \left\| \mathbf r_{i,m} - \operatorname{sg}(\mathbf q_{i,m}) \right\|_2^2
\big)
\Big)
\end{equation*}
where $\hat{\mathbf e}_i$ denotes the reconstructed representation, $\mathbf q_{i,m}$ and $\mathbf r_{i,m}$ denote the codebook embedding and the residual at stage $m$, $\operatorname{sg}(\cdot)$ is the stop-gradient operator, and $l$ is the length of the code sequence.

\subsection{Semantic-Guided Generation}
Although SID preserves the hierarchical semantic information of items, the downstream generative model re-embeds the code sequences during training for greater flexibility.
Since optimization primarily relies on historical interaction data, the learned embeddings mainly encode interaction and collaborative patterns, and the original semantic information obtained from pretrained language models will gradually be lost. Since the next item is expected to be decoded in the form of semantic code sequence, this information loss can lead to semantic distortion and reduced generation accuracy.
To address this, we introduce a semantic guidance module that explicitly restores semantic information during generation.

Since semantic representations are defined at the item level, we first transform hierarchical code embeddings into an item-level representation via a item-specific learnable vector $\mathbf{w} \in \mathbb{R}^{l}$.
To prevent the over-dominance of certain code positions and stabilize training, we optionally apply a temperature-scaled softmax to $\mathbf{w}$ to produce a more balanced weight distribution. Formally, let $\mathbf{c}_i = (c_{1i}, c_{2i}, \ldots, c_{li})$ denote the code sequence of item $i$. We first construct the learnable embedding matrix $\mathbf{E}_i \in \mathbb{R}^{l \times d}$, where the $k$-th row is the embedding of $c_{ki}$. The transformed code embedding $\mathbf{e}_i^{\text{code}} \in \mathbb{R}^{d}$ is then given by:
\[
\mathbf{w} = \mathbf{W}_h\mathbf{h}_i, \quad
\tilde{\mathbf{w}} = \operatorname{softmax}\!\left(\frac{\mathbf{w}}{\tau_w}\right), \quad
\mathbf{e}_i^{\text{code}} = \mathbf{E}_i^{\top} \tilde{\mathbf{w}},
\]
where $\mathbf{h}_i \in \mathbb{R}^{h}$ is the frozen original semantic representation, $\mathbf{W}_h \in \mathbb{R}^{l\times h}$, the temperature $\tau_w \geq 1$ and is linearly annealed to 1 during training.
$\mathbf{e}_i^{\text{code}}$ is primarily determined by the learnable code embedding matrix $\mathbf{E}$, which mainly captures collaborative and interaction patterns during generative model training. At this step, the semantic representation is only used to generate item-specific weights and does not alter the $d$-dimensional embedding space. The hierarchical cues of the code sequences can still be retained through learnable layer-specific weighting that is not globally shared.

Following this, we integrate the semantic representation with the learnable code embeddings dynamically:
\[
\begin{gathered}
\mathbf{e}^{\text{sem}}_i =  \mathbf{W}_s\mathbf{h}_i , \quad
g_i = \sigma\!\left(\mathbf{W}_g[\mathbf{e}_i^{\text{code}}; \mathbf{e}_i^{\text{sem}}]\right), \\
\mathbf{e}^{\text{fused}}_i =
g_i\,\mathbf{e}^{\text{code}}_i +
(1-g_i)\mathbf{e}^{\text{sem}}_i ,
\end{gathered}
\]
where $\mathbf{W}_s \in \mathbb{R}^{d\times h}$,
$\mathbf{W}_g \in \mathbb{R}^{1\times 2d}$, $\sigma$ is the sigmoid function, and $[\cdot;\cdot]$ denotes concatenation. This mechanism adaptively integrates collaborative representations with semantic priors according to the expressive capacity of the code embeddings. 
The computational cost of the whole procedure is approximately $2(ld + dh + lh)$ FLOPs up to lower-order terms, which does not introduce a noticeable overhead compared to the cost of the main generative model. The fusion form is interchangeable (e.g., dimension-wise or multi-head gating); here, we adopt the lowest-complexity design to reduce additional computational cost.

However, code and semantic embeddings originate from different modalities (behavioral interactions vs. textual semantics), making direct weighted fusion less effective. To capture cross-modal correspondence and facilitate downstream optimization, we introduce contrastive alignment between the two representations via the InfoNCE loss:
\[
\mathcal{L}_{align}
=
- \sum_{i \in \mathcal{I}}\log
\frac{
\exp(\mathrm{sim}(\mathbf{e}_i^\text{code}, \mathbf{e}_i^\text{sem}) / \tau)
}{
\sum_{j=1}^{N}
\exp(\mathrm{sim}(\mathbf{e}_i^\text{code}, \mathbf{e}_j^\text{sem}) / \tau)
}.
\]

With the fused representation of historical interaction items, 
we predict the code sequence of the next item using a backbone generative model and compute the corresponding generation loss $\mathcal{L}_{gen}$. The backbone model is flexible and compatible with a wide range of existing methods.


\subsection{Manifold Alignment}

While contrastive alignment mitigates the modality discrepancy between code and semantic representations at the instance level, a more intrinsic mismatch remains in their underlying geometry. Code embeddings originate from discrete and hierarchical codebook indices, whereas semantic representations correspond to the hidden states of a pretrained language model. As an intuitive analogy, directly fusing a language model’s token embedding matrix with its contextual hidden states, even if both originate from the same textual modality and are optimized under the same training objective, can still be suboptimal due to structural mismatch of representation spaces (e.g., differences in distance metrics and distributional structures). Solely aligning them with a contrastive loss based on inter-instance similarity  without resolving this discrepancy is therefore insufficient. 

To address this geometric mismatch, hyperbolic space is a valid shared latent space for aligning heterogeneous representation spaces, since it is a complete and simply connected negatively curved Riemannian manifold \cite{vinh2020hyperml}, and prior theoretical studies have shown its ability to better preserve intrinsic structural information compared to Euclidean alignment (e.g. minimizing $L_2$ distance or MMD loss)  \cite{nickel2017poincare}. As a typical example, Guo et al.~\cite{guo2021multi} project knowledge graph representations into hyperbolic space for entity alignment, and find that due to the exponential expansion property, it better preserves both hierarchical structure and similarity with less distortion, which are essential for modeling RQ-VAE-based code sequences and semantic representations, respectively. Motivated by the strong correspondence between these insights and our challenge, we further introduce hyperbolic manifold alignment to reconcile the underlying geometries of the two representation spaces.

Concretely, we project the code embedding $\mathbf{e}^{\text{code}}_{i}$ onto a $d$-dimensional Poincar\'e ball
\[
\mathbb{D}_c^d = \left\{ \mathbf{x} \in \mathbb{R}^d : \|\mathbf{x}\| < c \right\},
\]
via the exponential-map-style parametrization (following standard hyperbolic neural network formulations \cite{ganea2018hyperbolic}):
\[
\mathbf{z}_i^{\text{code}}
=
\phi(\mathbf{e}^{\text{code}}_{i})
=
\frac{c \tanh\!\left(\|\tilde{\mathbf e}^{\text{code}}_{i}\|/(2c)\right)}{\|\tilde{\mathbf e}^{\text{code}}_{i}\|}\,\tilde{\mathbf e}^{\text{code}}_{i} \in \mathbb{D}_c^d,
\]
where $c$ is the curvature hyperparameter,  $\tilde{\mathbf e}^{\text{code}}_i$ denotes the vector obtained by passing $\mathbf e^{\text{code}}_i$ through a learnable linear transformation. Similarly, we project the semantic representation $\mathbf{e}_i^{\mathrm{sem}}$ into the same hyperbolic space through a learnable projector $\psi(\cdot)$:
\[
\mathbf{z}_i^{\text{sem}} = \psi(\mathbf{e}_i^{\mathrm{sem}}) \in \mathbb{D}_c^d.
\]

We then minimize the geodesic distance between the two hyperbolic embeddings to enforce consistency between the two spaces:
\[
\begin{aligned}
\mathcal{L}_{\mathrm{geo}}
&=
\sum_{i \in \mathcal{I}}d_{\mathbb{D}}\!\left(\mathbf{z}_i^{\text{code}}, \mathbf{z}_i^{\text{sem}}\right)
\\
&=
\sum_{i \in \mathcal{I}}\cosh^{-1}\!\left(
1+
\frac{
2c^2 \left\| \mathbf{z}_i^{\text{code}} - \mathbf{z}_i^{\text{sem}} \right\|^2
}{
\left(c^2-\|\mathbf{z}_i^{\text{code}}\|^2\right)
\left(c^2-\|\mathbf{z}_i^{\text{sem}}\|^2\right)
}
\right),
\end{aligned}
\]
where $d_{\mathbb{D}}(\cdot,\cdot)$ is the geodesic distance on the Poincar\'e ball \cite{ganea2018hyperbolic,vinh2020hyperml}, and $c$ is the hyperparameter that controls the curvature of the hyperbolic space. The computations of $\phi(\cdot)$ and $d_{\mathbb{D}}(\cdot)$ are both linear in the embedding dimension, and the time complexity of computing $\mathcal{L}_{\mathrm{geo}}$ is $\mathcal{O}(Nd)$ for a batch of $N$ samples.  This is no higher than that of commonly used alignment objectives (e.g., a standard in-batch InfoNCE loss requires complexity $\mathcal{O}(N^2 d)$). Therefore, this module does not introduce noticeable computational cost.

While contrastive alignment facilitates model training by capturing cross-modal correspondences, this module further aligns the global geometry of the two representations within a shared latent hyperbolic space, yielding more coherent representations for generation.
Moreover, as we will later show empirically, reducing this underlying discrepancy also prevents the model from overly relying on semantic features, which are more easily exploited by the generative model as they correspond to the dense hidden states of the language model, thereby enabling more effective learning of code embeddings. We further provide a supplementary theoretical analysis showing that minimizing $L_{\mathrm{geo}}$ tightens an upper bound on the excess generation error induced by code embeddings under standard assumptions in Appendix~6.1.

The final loss function of the generation stage is given by:
\[
\mathcal{L} = \mathcal{L}_{gen} +\alpha\mathcal{L}_{align} + \beta\mathcal{L}_{geo}.
\]

Through the three interconnected modules proposed above, we ensures the consistent utilization of semantic and collaborative signals in both the tokenization and generation stages, mitigating the underlying information decoupling between the two stages.

\section{Experiments}
Our experiments address the following questions:
\begin{itemize}
    \item \textbf{RQ1:} 
    Does the proposed method consistently outperform baseline models across different datasets?

    \item \textbf{RQ2:} 
    Are all proposed modules in both the tokenization and generation stages individually effective and necessary for performance improvement?

    \item \textbf{RQ3:} 
    Is the proposed method robust and stable under different hyperparameter settings?

    \item \textbf{RQ4:} 
Does our method generalize effectively across different backbones with minimal additional overhead?

    \item \textbf{RQ5:} 
    Does explicitly injecting textualized collaborative signals during tokenization outperform alternative strategies?

\end{itemize}

\subsection{Setup}
\textbf{Dataset.} Following prior SID-based studies \cite{rajput2023tiger,hua2023index,jin2023language}, our experiments utilize three datasets from the Amazon Reviews \cite{mcauley2015image}: ``Beauty'', ``Sports and Outdoors'' and ``Toys and Games''. In alignment with established methodologies \cite{rajput2023tiger,zhou2020s3,hou2023learning}, we conceptualize a user's historical reviews as ``interactions.'' These interactions are sorted chronologically---from earliest to latest---to construct the input sequences. For our evaluation framework, we apply the standard leave-last-out protocol \cite{rajput2023tiger,kang2018sasrec,zhao2022revisiting}. Under this setup, the final item in a user's sequence is designated for testing, the penultimate item is held out for validation, and the remainder forms the training set. This is unified across all baselines to ensure fair comparison.

\textbf{Baselines.}
We include a diverse set of baselines, including recent and related representative methods across diverse generative paradigms to ensure a comprehensive evaluation:

\begin{itemize}

\item \textbf{Caser}~\cite{tang2018personalized} models sequential item ID patterns with CNNs to capture local dependencies.

\item \textbf{GRU4Rec}~\cite{hidasi2015session} uses a recurrent neural network for session-based recommendation.

\item \textbf{HGN}~\cite{ma2019hierarchical} extends RNN-based recommendation with gating mechanisms for improved sequential modeling.

\item \textbf{BERT4Rec}~\cite{sun2019bert4rec} applies a bidirectional Transformer with a Cloze-style objective to model item ID sequences.

\item \textbf{SASRec}~\cite{kang2018sasrec} employs a self-attention Transformer decoder for sequential recommendation.

\item \textbf{FDSA}~\cite{zhang2019feature} encodes item ID and feature sequences separately using self-attention.

\item \textbf{S$^3$-Rec}~\cite{zhou2020s3} performs self-supervised pretraining to capture correlations between item attributes and IDs before fine-tuning for next-item prediction.

\item \textbf{VQRec}~\cite{hou2023learning} converts items into semantic IDs via product quantization and represents items using pooled semantic ID embeddings.

\item \textbf{RecJPQ}~\cite{petrov2024recjpq} replaces item embeddings with shared sub-embeddings obtained through joint product quantization.

\item \textbf{HSTU}~\cite{zhai2024actions} discretizes raw item features into tokens as generative inputs.

\item \textbf{TIGER}~\cite{rajput2023tiger} tokenizes items into semantic IDs via RQ-VAE and predicts tokens autoregressively.

\item \textbf{LIGER}~\cite{yang2024unifying} extends TIGER by jointly supporting generative and dense retrieval through hybrid training and ranking.

\item \textbf{CoST}~\cite{zhu2024cost} introduces contrastive quantization to learn semantic tokens while preserving item neighborhood relations.

\item \textbf{ETEGRec}~\cite{liu2025generative} an end-to-end training framework with alternating optimization and recommendation-aware objectives.

\item \textbf{LC-Rec}~\cite{zheng2024adapting} leverages large language models for recommendation and enhances collaborative signals during generation.

\item \textbf{RPG}~\cite{hou2025generating} generates unordered token sequences via OPQ and accelerates inference using graph-based search.

\item \textbf{LLaDA}~\cite{shi2025llada} formulates SID generation as a discrete diffusion process with parallel token prediction and adaptive decoding.

\item \textbf{CoFiRec}~\cite{wei2025cofirec} adopts a coarse-to-fine tokenization framework to hierarchically generate semantic IDs.

\item \textbf{PIT}~\cite{wang2026pit} a very recent work that proposes a dynamic item tokenizer co-evolving with the generative recommender.

\end{itemize}

\textbf{Evaluation.} We adopt Recall@$K$ and NDCG@$K$ as evaluation metrics with $K \in \{5,10\}$, following Rajput et al.~\cite{rajput2023tiger}. 
The model achieving the best NDCG@$10$ on the validation set is selected for final evaluation on the test set.

\textbf{Implementation details.} 
We train the model for a maximum of 300 epochs with a batch size of 256 and a learning rate of 0.01, and employ early stopping with a patience of 20 based on validation NDCG@10. For the main experiments, we adopt a decoder-only backbone generative model following RPG ~\cite{hou2025generating}, and further evaluate our method on alternative backbone architectures in the subsequent section, where all hyperparameters and settings are kept consistent with the corresponding backbone. For the tokenization, we use the pretrained sentence-t5-base, which is commonly adopted in baselines, to extract item embeddings, and train the RQ-VAE for 5000 epochs using a learning rate of 0.001 and a batch size of 4096, the codebook size, latent space dimension, and the coefficient $\gamma$ are set to 256, 256, and 0.25, respectively.
We tune the loss weights $\alpha$ and $\beta$, the hyperbolic curvature $c$ in $\{0.2, 0.5, 1\}$, and $a$ in $\{0.2, 0.5, 0.8\}$ according to the validation $NDCG@10$. $\tau$ and $k$ are set to 0.1 and 10, respectively.
$\tau_w$ is linearly annealed from 2 to 1.
Other basic settings (optimizer, configurations for training and evaluation, data preprocessing and splitting, etc.) are kept consistent with RPG~\cite{hou2025generating}. 
Consistent with~\cite{lee2025gram}, the top-$10$ similar items are identified using SASRec pretrained using the training data. All experiments are conducted on a single RTX 4090 GPU. Baselines including RPG are rerun under the same environment using the official open-source implementation and strictly following the instructions, both TIGER and LIGER are evaluated using LIGER's open-source codebase.


\subsection{Overall Performance}
We compare our method with a diverse set of baselines on three widely used datasets. As shown in Table~\ref{tab:main_exp}, our method consistently outperforms all baseline methods, achieving 13\%–27\% gains over the second-best method across four metrics. These results demonstrate strong performance and robustness of our method, providing a clear answer to RQ1. 

Among the baselines, semantic ID–based methods significantly outperform item ID–based ones, highlighting the importance of semantic tokenization. However, due to the decoupling of semantic and collaborative signals across tokenization and generation, their performance remains suboptimal.

\begin{table*}[t]
\centering
\small
\caption{Performance comparison among baselines and the proposed method. The best performance score is denoted in bold.}
\label{tab:main_exp}
\resizebox{\textwidth}{!}{
\begin{tabular}{lcccccccccccc}
\toprule
\multirow{2}{*}{Model} &
\multicolumn{4}{c}{Beauty} &
\multicolumn{4}{c}{Sports and Outdoors} &
\multicolumn{4}{c}{Toys and Games} \\
\cmidrule(lr){2-5}\cmidrule(lr){6-9}\cmidrule(lr){10-13}
& R@5 & N@5 & R@10 & N@10
& R@5 & N@5 & R@10 & N@10
& R@5 & N@5 & R@10 & N@10 \\
\midrule

\multicolumn{13}{c}{\textit{Item ID-based}} \\
\midrule

Caser   & 0.0205 & 0.0131 & 0.0347 & 0.0176 & 0.0116 & 0.0072 & 0.0194 & 0.0097 & 0.0166 & 0.0107 & 0.0270 & 0.0141 \\
GRU4Rec & 0.0164 & 0.0099 & 0.0283 & 0.0137 & 0.0129 & 0.0086 & 0.0204 & 0.0110 & 0.0097 & 0.0059 & 0.0176 & 0.0084 \\
HGN     & 0.0325 & 0.0206 & 0.0512 & 0.0266 & 0.0189 & 0.0120 & 0.0313 & 0.0159 & 0.0321 & 0.0221 & 0.0497 & 0.0277 \\
BERT4Rec& 0.0203 & 0.0124 & 0.0347 & 0.0170 & 0.0115 & 0.0075 & 0.0191 & 0.0099 & 0.0116 & 0.0071 & 0.0203 & 0.0099 \\
SASRec  & 0.0387 & 0.0249 & 0.0605 & 0.0318 & 0.0233 & 0.0154 & 0.0350 & 0.0192 & 0.0463 & 0.0306 & 0.0675 & 0.0374 \\
FDSA    & 0.0267 & 0.0163 & 0.0407 & 0.0208 & 0.0182 & 0.0122 & 0.0288 & 0.0156 & 0.0228 & 0.0140 & 0.0381 & 0.0189 \\
S$^3$-Rec & 0.0387 & 0.0244 & 0.0647 & 0.0327 & 0.0251 & 0.0161 & 0.0385 & 0.0204 & 0.0443 & 0.0294 & 0.0700 & 0.0376 \\

\midrule
\multicolumn{13}{c}{\textit{Semantic ID-based}} \\
\midrule

RecJPQ & 0.0311 & 0.0167 & 0.0482 & 0.0222 & 0.0141 & 0.0076 & 0.0220 & 0.0102 & 0.0331 & 0.0182 & 0.0484 & 0.0231 \\
VQ-Rec & 0.0457 & 0.0317 & 0.0664 & 0.0383 & 0.0208 & 0.0144 & 0.0300 & 0.0173 & 0.0497 & 0.0346 & 0.0737 & 0.0423 \\
HSTU   & 0.0469 & 0.0314 & 0.0704 & 0.0389 & 0.0258 & 0.0165 & 0.0414 & 0.0215 & 0.0433 & 0.0281 & 0.0669 & 0.0357 \\
Tiger  & 0.0378 & 0.0243 & 0.0592 & 0.0312 & 0.0230 & 0.0146 & 0.0346 & 0.0184 & 0.0381 & 0.0245 & 0.0594 & 0.0313 \\
Liger  & 0.0456 & 0.0291 & 0.0697 & 0.0368 & 0.0281 & 0.0184 & 0.0419 & 0.0228 & 0.0474 & 0.0296 & 0.0704 & 0.0370 \\
CoST & 0.0446 & 0.0298 & 0.0688 & 0.0375 & 0.0246 & 0.0158 & 0.0398 & 0.0207 & 0.0458 & 0.0311 & 0.0683 & 0.0383 \\
ETEGRec & 0.0408 & 0.0270 & 0.0657 & 0.0344 & 0.0269 & 0.0164 & 0.0435 & 0.0221 & 0.0391 & 0.0245 & 0.0602 & 0.0321 \\
LC-Rec  & 0.0433 & 0.0289 & 0.0653 & 0.0360 & 0.0271 & 0.0177 & 0.0428 & 0.0228 & 0.0466 & 0.0323 & 0.0687 & 0.0394 \\
RPG &
\underline{0.0547} & \underline{0.0380} & \underline{0.0795} & \underline{0.0460} &
\underline{0.0299} & \underline{0.0207} & 0.0440 & \underline{0.0252} &
\underline{0.0576} & \underline{0.0394} & \underline{0.0857} & \underline{0.0484} \\
LLaDA  & 0.0435 & \underline{0.0306} & 0.0620 & 0.0364 & 0.0269 & 0.0179 & 0.0422 & 0.0228 & 0.0374 & 0.0263 & 0.0543 & 0.0318 \\
CoFiRec& 0.0444 & 0.0295 & 0.0679 & 0.0370 & \underline{0.0290} & 0.0180 & \underline{0.0456} & \underline{0.0245} & 0.0480 & 0.0307 & 0.0785 & 0.0405 \\
S²GR & 0.0465 & 0.0301 & 0.0755 & 0.0401 & 0.0282 & 0.0183 & 0.0451 & 0.0238 & 0.0505 & 0.0322 & 0.0821 & 0.0423 \\
PIT     & \underline{0.0479} & 0.0305 & \underline{0.0783} & \underline{0.0403} & 0.0284 & \underline{0.0186} & \underline{0.0462} & 0.0243 & \underline{0.0516} & \underline{0.0328} & \underline{0.0839} & \underline{0.0432} \\

\midrule
\textbf{Ours} & \textbf{0.0644} & \textbf{0.0454} & \textbf{0.0930} & \textbf{0.0546} & \textbf{0.0380} & \textbf{0.0258} & \textbf{0.0570} & \textbf{0.0319} & \textbf{0.0709} & \textbf{0.0497} & \textbf{0.0973} & \textbf{0.0582} \\
Rel. Impr. $\uparrow$ &
17.73\% & 19.47\% & 16.98\% & 18.70\% &
27.09\% & 24.64\% & 23.38\% & 26.59\% &
23.09\% & 26.14\% & 13.54\% & 20.25\% \\

\bottomrule
\end{tabular}
}
\end{table*}

\subsection{Ablation Study}
We conduct ablation studies to evaluate each component of our framework. As shown in Table~\ref{tab:ablation}, removing collaborative signals in tokenization (i.e., $a=1$) or removing semantic fusion in generation (i.e., $g=1$ and excluding additional alignment modules) both lead to clear performance drops, validating the importance of bidirectional information supplementation.
Furthermore, removing either contrastive alignment or manifold alignment in the generation stage consistently degrades performance. Notably, even with optimally weighted contrastive loss, manifold alignment still yields significant gains, highlighting the necessity of reconciling the underlying geometries of the two representation spaces. 

These results are consistent across all datasets and metrics, confirming the robustness of our findings and providing a comprehensive answer to RQ2.

\begin{table*}[t]
\centering
\caption {Results of ablation study. Rows 2–5 report the evaluation results after ablating collaborative tokenization, semantic guidance, manifold alignment, and contrastive alignment, respectively. When semantic guidance is ablated, the alignment modules in the last two rows are also removed since only one representation remains.}
\label{tab:ablation}
\small
\begin{tabular}{lcccccccccccc}
\toprule
\multirow{2}{*}{Model}
& \multicolumn{4}{c}{Beauty}
& \multicolumn{4}{c}{Sports and Outdoors}
& \multicolumn{4}{c}{Toys and Games} \\
\cmidrule(lr){2-5} \cmidrule(lr){6-9} \cmidrule(lr){10-13}
& R@5 & N@5 & R@10 & N@10
& R@5 & N@5 & R@10 & N@10
& R@5 & N@5 & R@10 & N@10 \\
\midrule
Ours
& \textbf{0.0644} & \textbf{0.0454} & \textbf{0.0930} & \textbf{0.0546} & \textbf{0.0380} & \textbf{0.0258} & \textbf{0.0570} & \textbf{0.0319} & \textbf{0.0709} & \textbf{0.0497} & \textbf{0.0973} & \textbf{0.0582} \\

Ours w/o. CT
& 0.0584 & 0.0402 & 0.0896 & 0.0502
& 0.0314 & 0.0229 & 0.0460 & 0.0275
& 0.0621 & 0.0442 & 0.0922 & 0.0538 \\

Ours w/o. SG
& 0.0524 & 0.0369 & 0.0757 & 0.0444
& 0.0284 & 0.0188 & 0.0439 & 0.0237
& 0.0556 & 0.0386 & 0.0795 & 0.0463 \\

Ours w/o. MA
& 0.0570 & 0.0402 & 0.0834 & 0.0491
& 0.0351 & 0.0241 & 0.0525 & 0.0297
& 0.0666 & 0.0464 & 0.0926 & 0.0547 \\

Ours w/o. CA
& 0.0605 & 0.0423 & 0.0870 & 0.0509
& 0.0356 & 0.0248 & 0.0530 & 0.0305
& 0.0664 & 0.0476 & 0.0923 & 0.0558 \\

\bottomrule
\end{tabular}
\end{table*}

\subsection{Sensitivity Analysis}
To examine the robustness of our method, we conduct extensive sweeps over the loss weights $\alpha$, $\beta$, and the hyperbolic curvature $c$ across all datasets. As shown in Figure~\ref{fig:sens}, the performance of our method remains stable under different settings and consistently outperforms all baselines on all datasets and metrics. These results confirm the stability and effectiveness of our approach, providing a clear answer to RQ3.

\begin{figure}
    \centering
    \includegraphics[width=\linewidth]{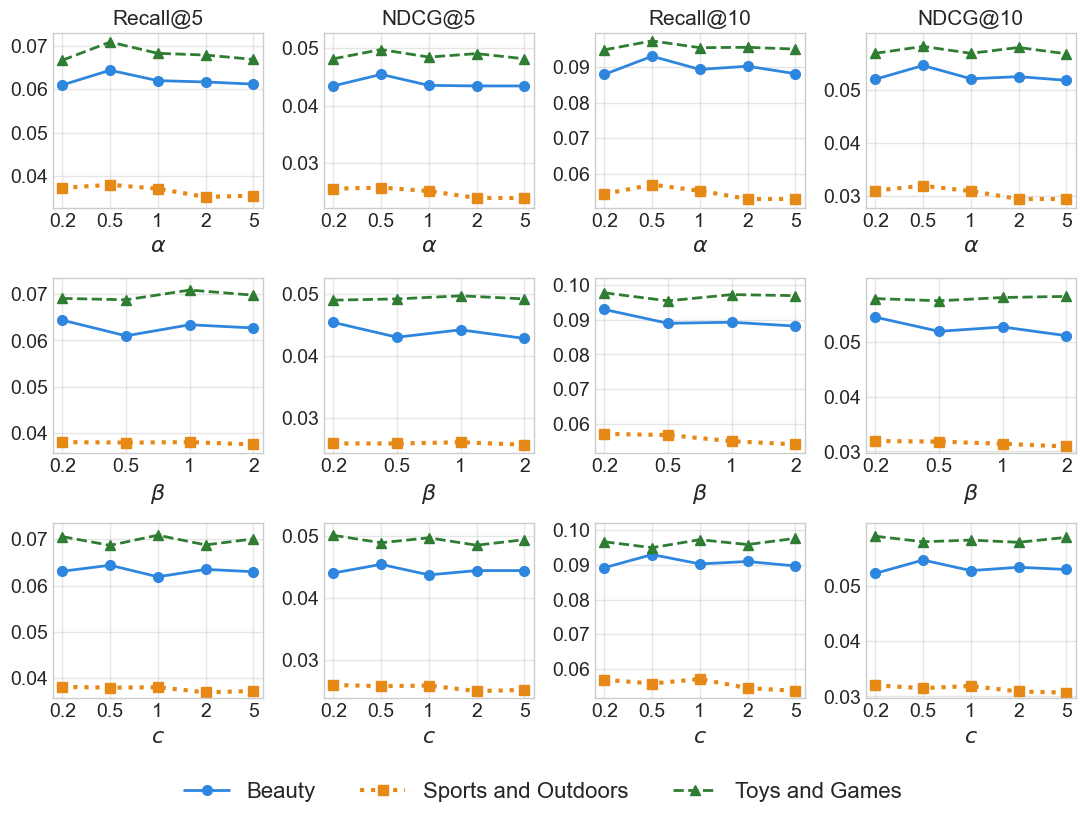}
    \caption{Hyperparameter Sensitivity Analysis.}
    \label{fig:sens}
\end{figure}


\subsection{Generalizability and Efficiency}
Our proposed framework is lightweight and agnostic to specific quantization or generation methods. To demonstrate its generalizability, we implement it as a reusable module and evaluate it on different backbones. Specifically, we adopt the representative method TIGER (RQ-VAE + T5), and LIGER (a recent and representative two-stage framework unifying generative and dense retrieval). We compare each backbone with and without our module under identical settings. As shown in Table~\ref{tab:plug}, our module consistently improves performance across all datasets, demonstrating strong generalizability and answering RQ4.

Moreover, our module does not significantly increase the cost of online training or inference. On Beauty dataset, after integrating our module, the per-step training time and inference time (averaged over 10 runs) of TIGER increase by only 6.70\% and 6.19\%, respectively, while the corresponding increases for LIGER are merely 1.98\% and 0.54\%. As demonstrated in Figure ~\ref{fig:plug}, this is minimal compared to the performance gains, further confirming the strong generalizability and deployment potential of our framework.
\begin{figure}[htbp]
    \centering
    \begin{subfigure}{0.495\linewidth}
        \centering
        \includegraphics[width=\linewidth]{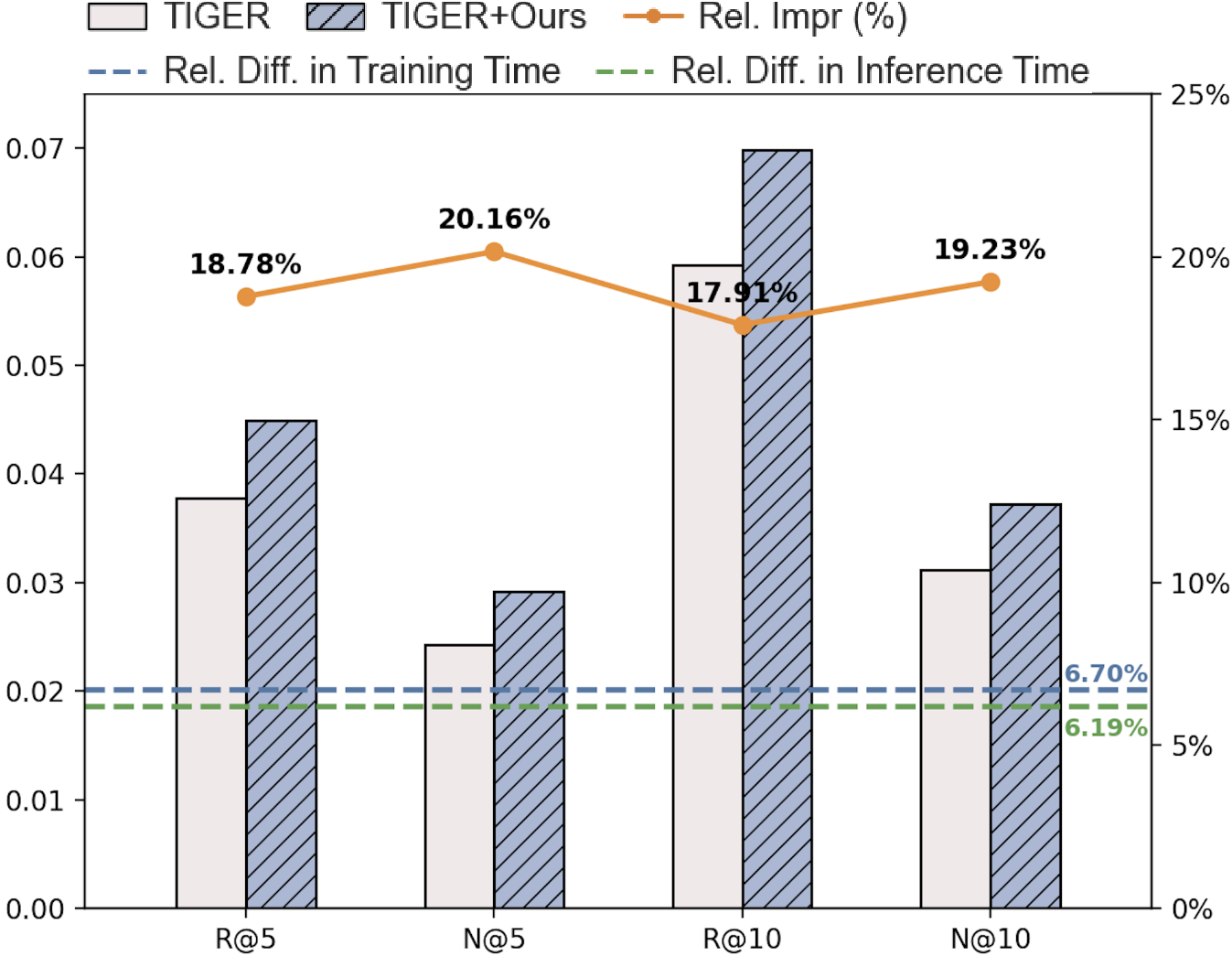}
        \caption{Tiger}
    \end{subfigure}
    \hfill
    \begin{subfigure}{0.495\linewidth}
        \centering
        \includegraphics[width=\linewidth]{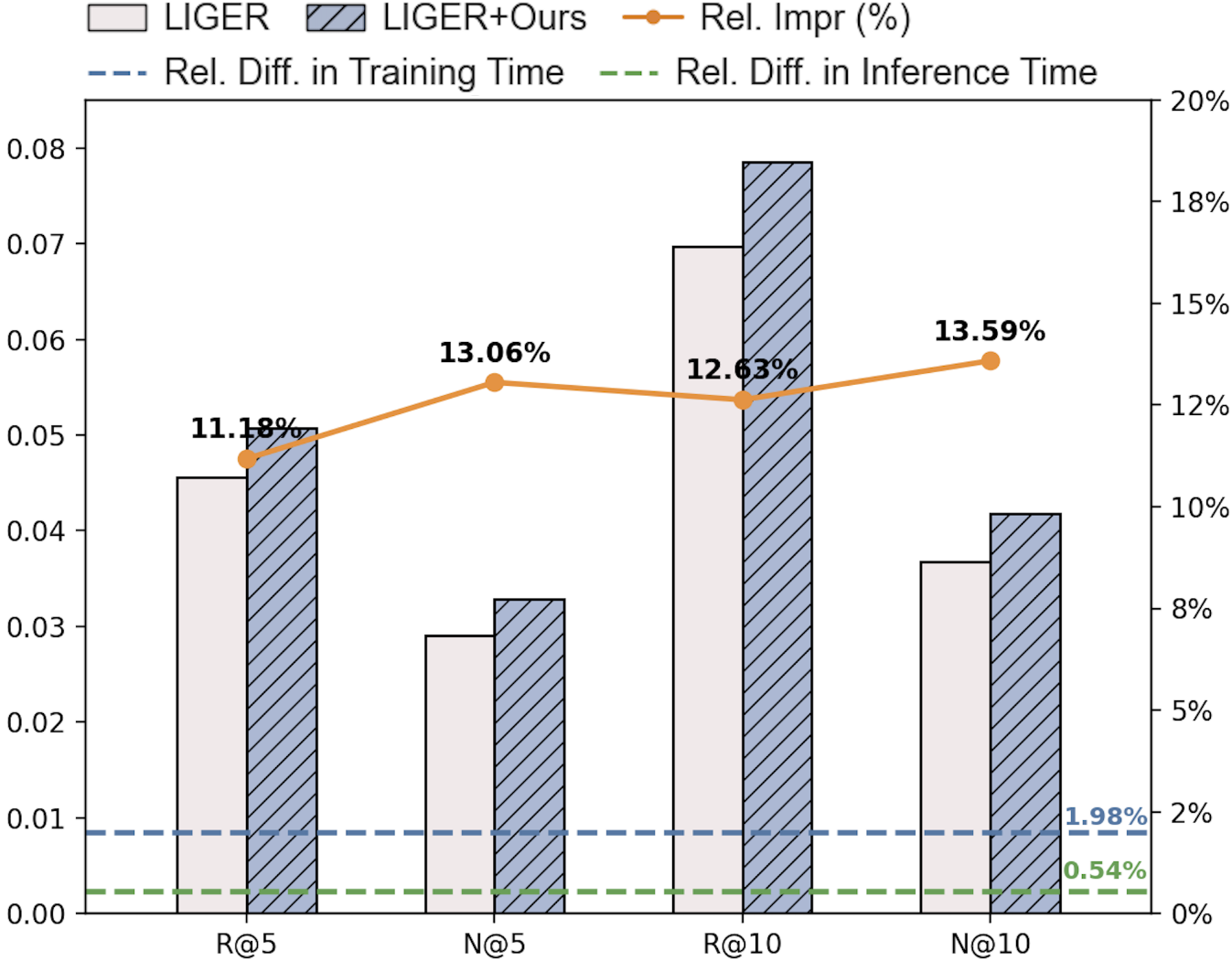}
        \caption{Liger}
    \end{subfigure}

    \caption{Effect of our module across different backbones.}
    \label{fig:plug}
\end{figure}

\begin{table*}[t]
\centering
\caption{The performance of our general module across different backbone models.}
\label{tab:plug}
\small
\begin{tabular}{lcccccccccccc}
\toprule
\multirow{2}{*}{Model}
& \multicolumn{4}{c}{Beauty}
& \multicolumn{4}{c}{Sports and Outdoors}
& \multicolumn{4}{c}{Toys and Games} \\
\cmidrule(lr){2-5} \cmidrule(lr){6-9} \cmidrule(lr){10-13}
& R@5 & N@5 & R@10 & N@10
& R@5 & N@5 & R@10 & N@10
& R@5 & N@5 & R@10 & N@10 \\
\midrule
TIGER
& 0.0378 & 0.0243 & 0.0592 & 0.0312
& 0.0230 & 0.0146 & 0.0346 & 0.0184
& 0.0381 & 0.0245 & 0.0594 & 0.0313 \\

TIGER+Ours
& \textbf{0.0449} & \textbf{0.0292} & \textbf{0.0698} & \textbf{0.0372}
& \textbf{0.0259} & \textbf{0.0170} & \textbf{0.0396} & \textbf{0.0213}
& \textbf{0.0442} & \textbf{0.0284} & \textbf{0.0697} & \textbf{0.0366} \\

Rel. Impr.$\uparrow$
& 18.78\% & 20.16\% & 17.91\% & 19.23\%
& 12.61\% & 16.44\% & 14.45\% & 15.76\%
& 16.01\% & 15.92\% & 17.34\% & 16.93\% \\

\midrule
LIGER
& 0.0456 & 0.0291 & 0.0697 & 0.0368 & 0.0281 & 0.0184 & 0.0419 & 0.0228 & 0.0474 & 0.0296 & 0.0704 & 0.0370 \\

LIGER+Ours
& \textbf{0.0507} & \textbf{0.0329} & \textbf{0.0785} & \textbf{0.0418}
& \textbf{0.0314} & \textbf{0.0205} & \textbf{0.0460} & \textbf{0.0252}
& \textbf{0.0536} & \textbf{0.0336} & \textbf{0.0786} & \textbf{0.0417} \\

Rel. Impr.$\uparrow$
& 11.18\% & 13.06\% & 12.63\% & 13.59\% & 11.74\% & 11.41\% & 9.79\% & 10.53\% & 13.08\% & 13.51\% & 11.65\% & 12.70\% \\

\bottomrule
\end{tabular}
\end{table*}

\subsection{Further Analysis}
\subsubsection{Comparison with Alternative Collaborative Tokenization Strategies}
As discussed, we enhance consistency between tokenization and downstream generation by explicitly incorporating textualized collaborative signals. While ablation studies verify its necessity, we further compare with representative alternatives of comparable complexity (the model for CF embeddings extraction is unified for fair comparison):
\begin{itemize} 
\item Following MMGRec \cite{liu2024mmgrec}, CF embeddings are concatenated with semantic embeddings before quantization. 
\item Following LETTER \cite{wang2024letter}, collaborative signals are introduced via contrastive learning between CF embeddings and the reconstructed semantic embeddings in RQ-VAE training. \end{itemize}
All other components (including the generation model and alignment module) remain unchanged. 

As shown in Table~\ref{tab:alter}, replacing our collaborative fusion strategy with these alternatives leads to noticeable performance degradation. The second alternative method underperforms slightly compared to removing collaborative signals on the latter two datasets, which is consistent with the results in \cite{lepage2025closing}. As discussed, the first alternative may distort the representation space due to modality mismatch; the second takes only semantic embeddings as inputs, such an implicit fusion mechanism cannot ensure effective optimization of the auxiliary objectives or sufficient capture of collaborative signal, and may interfere with the optimization of the original quantization reconstruction objective. In contrast, our approach adopts a simple yet effective strategy that explicitly integrates semantic and collaborative signals during tokenization. The two representations are naturally aligned since they are extracted by the same language model. These results further validate the advantage of our design and address RQ5.

\begin{table}[t]
\centering
\caption{Comparison with alternative strategies for collaborative signal fusion in tokenization (NDCG@10).}

\label{tab:alter}
\resizebox{1\linewidth}{!}{
\begin{tabular}{lccc}
\toprule
 & Beauty & Sports and Outdoors & Toys and Games \\
\midrule
MMGRec based & 0.0512 & 0.0280 & 0.0557 \\
LETTER based & 0.0523 & 0.0258 & 0.0526 \\
Ours & \textbf{0.0546} & \textbf{0.0319} & \textbf{0.0582} \\
w/o. CT & 0.0502 & 0.0275 & 0.0538 \\
\bottomrule
\end{tabular}
}
\end{table}

\subsubsection{Comparison with Other Alignment Methods}
To validate the necessity of the Manifold Alignment module, we replaced it with commonly used Euclidean $L_2$ distance minimization and MMD loss, while keeping all other modules unchanged. As shown in Table~\ref{tab:otherrep}, these alternatives perform noticeably worse than our proposed method (with one result falling below that of the ablation variant). These results confirm the superiority of aligning the two representations in a shared hyperbolic latent space, which preserves both hierarchical structure and similarity.

\begin{table}[t]
\centering
\caption{Comparison with alternative methods for representation alignment (NDCG@10).}

\label{tab:otherrep}
\resizebox{1\linewidth}{!}{
\begin{tabular}{lccc}
\toprule
 & Beauty & Sports and Outdoors & Toys and Games \\
\midrule
$L_2$ distance & 0.0517 & 0.0285 & 0.0551 \\
MMD loss & 0.0504 & 0.0306 & 0.0573 \\
Ours & \textbf{0.0546} & \textbf{0.0319} & \textbf{0.0582} \\
w/o. MA & 0.0487 & 0.0297 & 0.0547 \\
\bottomrule
\end{tabular}
}
\end{table}

\subsubsection{Visualization}
For a more fine-grained analysis of the role of manifold alignment, we visualize the code and semantic representations in a two-dimensional space using t-SNE. Taking the Beauty dataset as an example, Figure~\ref{fig:tsne} compares the visualizations before and after manifold alignment. We observe that aligning the two representations in a shared hyperbolic space brings their distributions significantly closer in the projected space.
Note that this visualization is obtained after incorporating contrastive alignment based on inter-instance similarity; the additional improvement brought by manifold alignment further demonstrates that reconciling the underlying geometries of the two representation spaces is effective for learning coherent representations for generation.

Furthermore, we analyze the distribution of the gating weight $g$ with and without manifold alignment. As shown in Figure~\ref{fig:hist}, $g$ is noticeably smaller without alignment, with a substantial portion close to $0$, while it becomes more symmetric and concentrates around $0.5$ with alignment. This is likely because the semantic representation corresponds to the dense hidden states of a pretrained language model, which are easier for the generative model to exploit, whereas the code sequences are newly embedded.
As a result, the model tends to rely more on the semantic representation and assigns uneven weights during fusion, especially in early training. By reducing the underlying discrepancy between the two spaces, manifold alignment encourages more balanced weight allocation, enabling more effective learning of code embeddings and better preserving the information encoded in the code sequences, thereby strengthening cross-stage information transfer.
\begin{figure}[htbp]
    \centering
    \begin{subfigure}{0.49\linewidth}
        \centering
        \includegraphics[width=\linewidth]{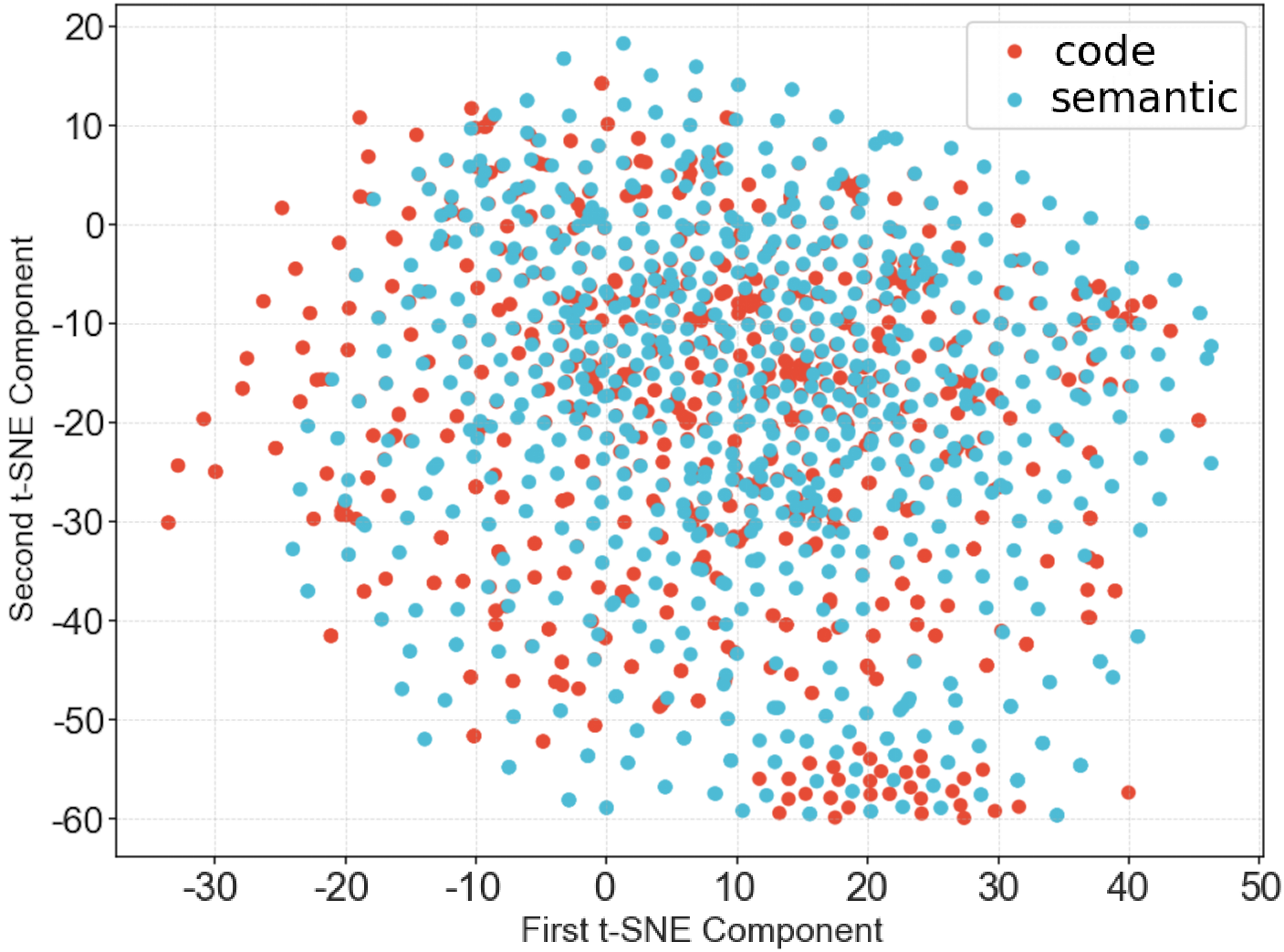}
        \caption{with manifold alignment}
    \end{subfigure}
    \hfill
    \begin{subfigure}{0.49\linewidth}
        \centering
        \includegraphics[width=\linewidth]{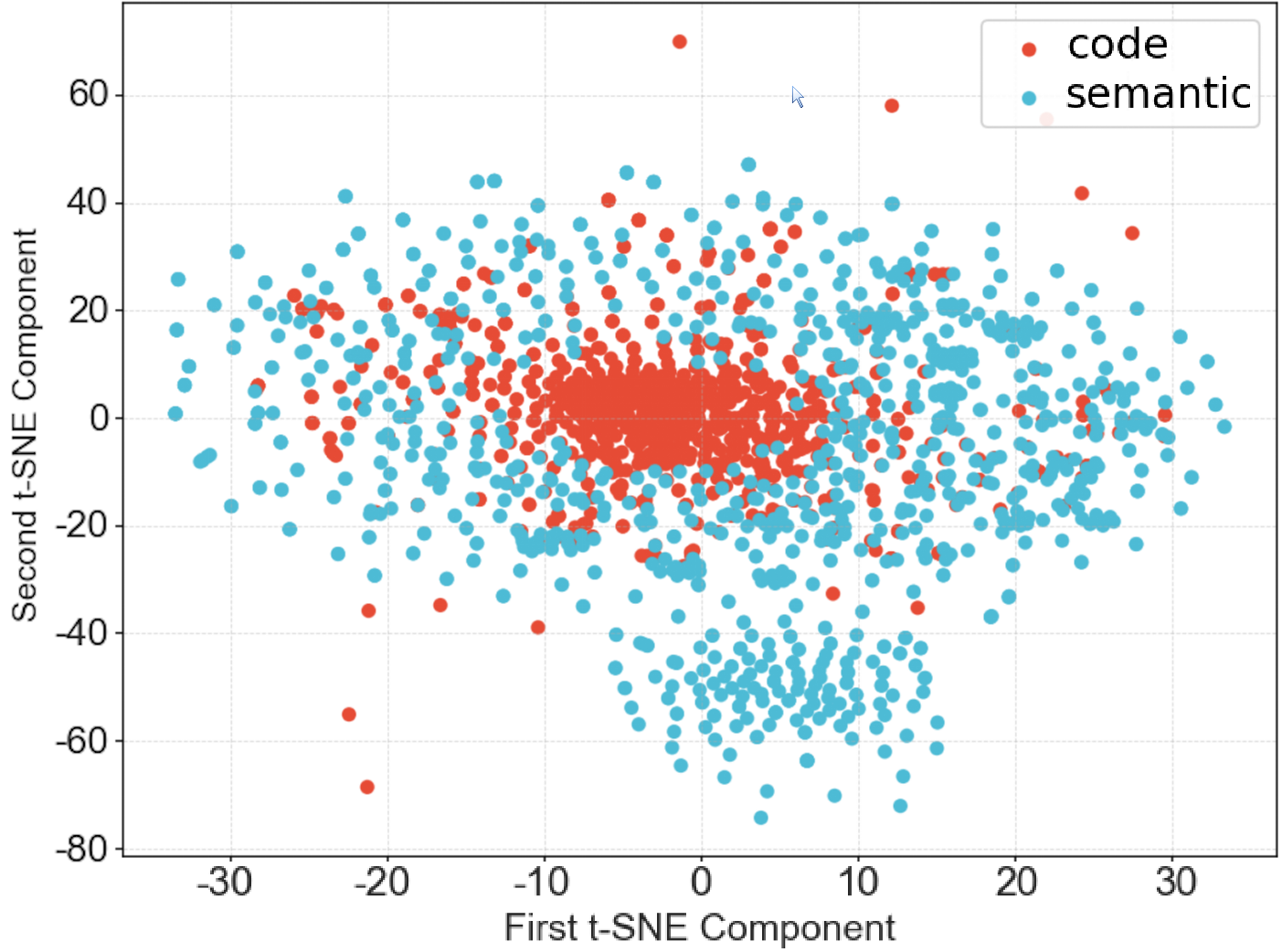}
        \caption{without manifold alignment}
    \end{subfigure}

    \caption{2D visualization of two representations.}
    \label{fig:tsne}
\end{figure}


\begin{figure}[htbp]
    \centering
    \begin{subfigure}{0.49\linewidth}
        \centering
        \includegraphics[width=\linewidth]{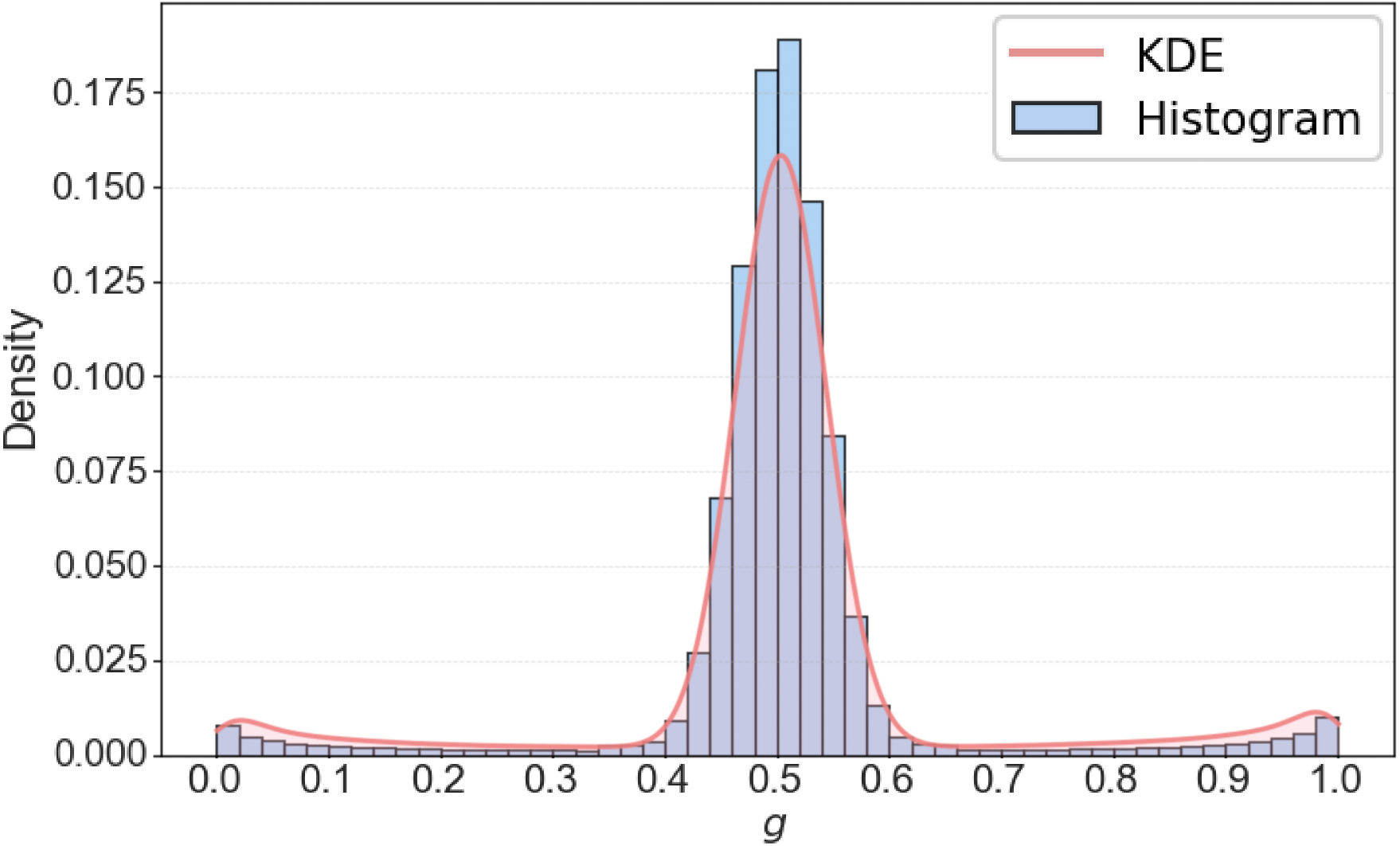}
        \caption{with manifold alignment}
    \end{subfigure}
    \hfill
    \begin{subfigure}{0.49\linewidth}
        \centering
        \includegraphics[width=\linewidth]{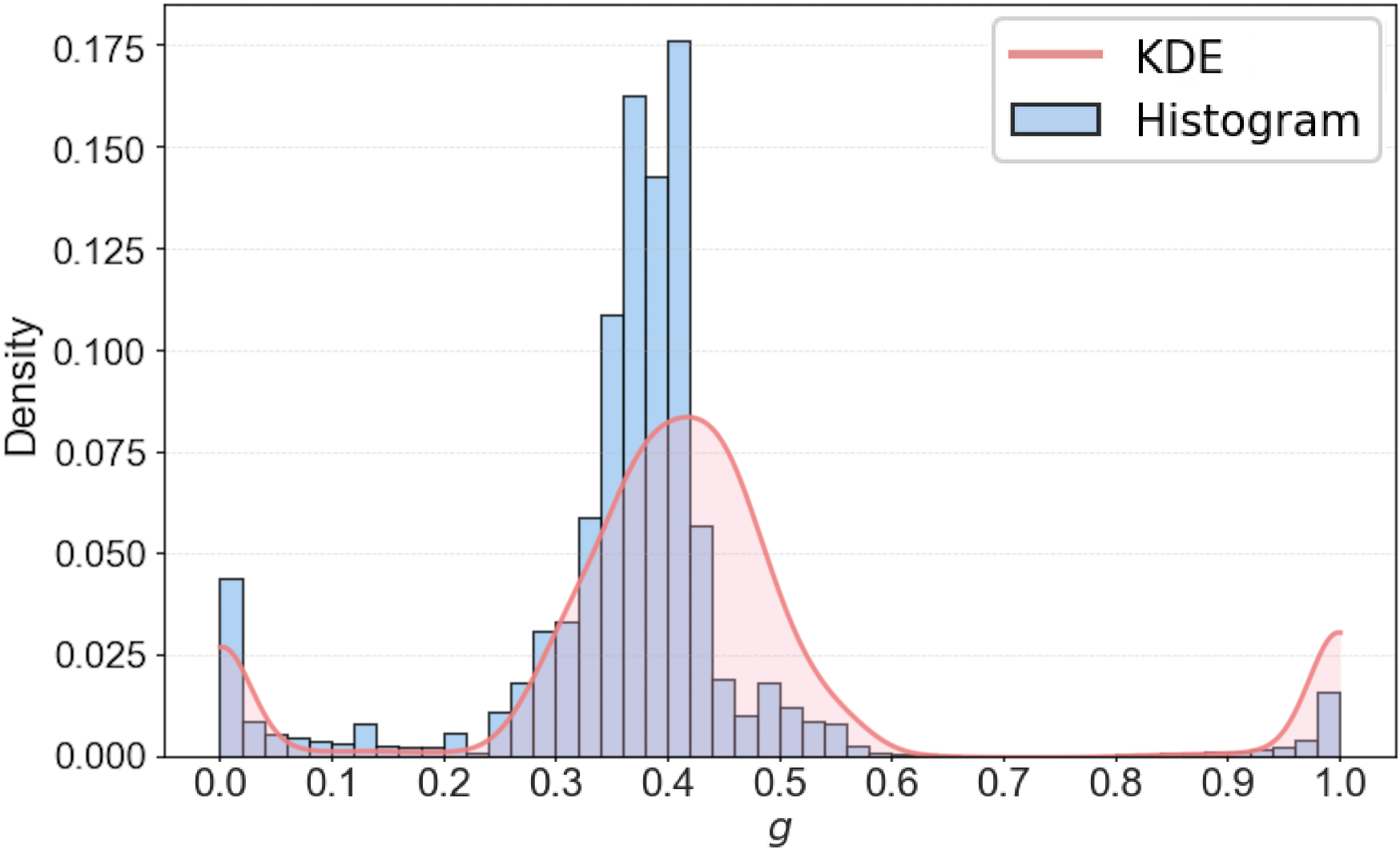}
        \caption{without manifold alignment}
    \end{subfigure}

    \caption{Distribution of gating weight $g$.}
    \label{fig:hist}
\end{figure}

\section{Conclusion}
In this work, we propose a general framework to address the cross-stage decoupling between semantic and collaborative signals in generative recommendation. Specifically, we inject textualized collaborative signals into semantic tokenization explicitly, recover semantic priors during the generation stage, and further align the underlying geometries of the two representations in a shared hyperbolic space.
Extensive experiments demonstrate that the proposed method consistently outperforms a diverse set of baseline models across multiple datasets and metrics. Comprehensive ablation studies and hyperparameter analyses further validate the effectiveness of each component and the robustness of our method.
We also implement the proposed framework as a reusable module and show that it achieves significant performance improvements across multiple backbone models with marginal additional computational overhead. These results further highlight the broad applicability and practical value of our approach.

\bibliographystyle{ACM-Reference-Format}
\bibliography{references}

\appendix

\section{Appendix}

\subsection{Theoretical Analysis of Manifold Alignment}

\begin{proposition}
Let $e_i^{\mathrm{code}} \in \mathbb{R}^d$ and $e_i^{\mathrm{sem}} \in \mathbb{R}^d$ denote the code-induced and semantic representations of item $i$, and let
\[
z_i^{c}=\phi(e_i^{\mathrm{code}})\in \mathbb{D}_c^d,\qquad
z_i^{s}=\psi(e_i^{\mathrm{sem}})\in \mathbb{D}_c^d
\]
be their projections in the shared Poincar\'e ball. Assume there exists an ideal latent representation $z_i^\star\in\mathbb{D}_c^d$ such that the conditional distribution of the target next-item code sequence $y_i$ depends on $i$ through $z_i^\star$.

Let $f_\theta:\mathbb{D}_c^d\to\mathbb{R}^m$ be the decoder-logit map and let $\ell(\cdot,y)$ be the token-level generation loss. Suppose that:

(i) there exists $L_f>0$ such that
\[
\|f_\theta(u)-f_\theta(v)\|_2 \le L_f\, d_{\mathbb D}(u,v),
\qquad \forall u,v\in\mathbb{D}_c^d;
\]

(ii) for any fixed target $y$, $\ell(\cdot,y)$ is $L_\ell$-Lipschitz, i.e.,
\[
|\ell(a,y)-\ell(b,y)|\le L_\ell \|a-b\|_2,
\qquad \forall a,b\in\mathbb{R}^m;
\]

(iii) the semantic representation satisfies
$
\mathbb{E}\!\left[d_{\mathbb D}(z_i^{s},z_i^\star)\right]\le \varepsilon_s.
$

Then
\[
\mathbb{E}\!\left[\ell(f_\theta(z_i^{c}),y_i)-\ell(f_\theta(z_i^\star),y_i)\right]
\le
L_\ell L_f
\left(
\mathbb{E}[d_{\mathbb D}(z_i^{c},z_i^{s})]+\varepsilon_s
\right).
\]
Therefore, minimizing the manifold alignment loss
\[
L_{\mathrm{geo}}=\sum_{i\in\mathcal I} d_{\mathbb D}(z_i^{c},z_i^{s})
\]
tightens an upper bound on the excess generation error induced by the re-embedded code representation relative to ideal latent representation.
\end{proposition}

\begin{proof}
For any item $i$, by the Lipschitz continuity of $\ell(\cdot,y_i)$,
\begin{align*}
\ell(f_\theta(z_i^{c}),y_i)-\ell(f_\theta(z_i^\star),y_i)
&\le
\left|\ell(f_\theta(z_i^{c}),y_i)-\ell(f_\theta(z_i^\star),y_i)\right| \\
&\le
L_\ell \|f_\theta(z_i^{c})-f_\theta(z_i^\star)\|_2.
\end{align*}
By the Lipschitz continuity of $f_\theta$ with respect to the hyperbolic distance,
\[
\|f_\theta(z_i^{c})-f_\theta(z_i^\star)\|_2
\le
L_f\, d_{\mathbb D}(z_i^{c},z_i^\star).
\]
Hence,
\[
\ell(f_\theta(z_i^{c}),y_i)-\ell(f_\theta(z_i^\star),y_i)
\le
L_\ell L_f\, d_{\mathbb D}(z_i^{c},z_i^\star).
\]
Since $d_{\mathbb D}$ is a metric on the Poincar\'e ball, it satisfies the triangle inequality:
\[
d_{\mathbb D}(z_i^{c},z_i^\star)
\le
d_{\mathbb D}(z_i^{c},z_i^{s}) + d_{\mathbb D}(z_i^{s},z_i^\star).
\]
Therefore,
\[
\ell(f_\theta(z_i^{c}),y_i)-\ell(f_\theta(z_i^\star),y_i)
\le
L_\ell L_f
\Big(
d_{\mathbb D}(z_i^{c},z_i^{s}) + d_{\mathbb D}(z_i^{s},z_i^\star)
\Big).
\]
Taking expectation on both sides gives
\[
\mathbb{E}\!\left[\ell(f_\theta(z_i^{c}),y_i)-\ell(f_\theta(z_i^\star),y_i)\right]
\le
L_\ell L_f
\left(
\mathbb{E}[d_{\mathbb D}(z_i^{c},z_i^{s})]
+
\mathbb{E}[d_{\mathbb D}(z_i^{s},z_i^\star)]
\right).
\]
Applying assumption (iii), we obtain
\[
\mathbb{E}\!\left[\ell(f_\theta(z_i^{c}),y_i)-\ell(f_\theta(z_i^\star),y_i)\right]
\le
L_\ell L_f
\left(
\mathbb{E}[d_{\mathbb D}(z_i^{c},z_i^{s})]+\varepsilon_s
\right).
\]
This completes the proof.
\end{proof}

\begin{remark}
Assumption (i) is a standard smoothness condition on the predictor. Similar Lipschitz assumptions are widely adopted in prior studies ~\cite{anil2019sorting}.
Assumption (ii) is also standard in statistical learning theory where Lipschitz losses are routinely assumed~\cite{hou2023instance}.
Assumption (iii) can be viewed as a bounded-error proxy assumption, where the semantic representation approximates the ideal latent representation on average, without requiring exact matching.
Similar structural assumptions are common in prior studies where representations are assumed to share a common latent cause~\cite{wang2015deep}.
\end{remark}

\end{document}